\documentclass{article}

\PassOptionsToPackage{numbers, compress}{natbib}

\usepackage[final]{neurips_2024}

\usepackage[utf8]{inputenc} 
\usepackage[T1]{fontenc}    
\usepackage[colorlinks,allcolors=teal]{hyperref}       
\usepackage{url}            
\usepackage{booktabs}       
\usepackage{amsfonts}       
\usepackage{nicefrac}       
\usepackage{microtype}      
\usepackage{xcolor}         

\usepackage{mathtools}
\usepackage{cleveref}
\usepackage{algorithm}
\usepackage[noend]{algpseudocode}
\usepackage{amsthm}
\usepackage{subcaption}
\newtheorem{theorem}{Theorem}
\newtheorem{proposition}{Proposition}
\theoremstyle{definition}

\newcommand{\emdash}{\,---\,}

\newcommand{\feds}{\textsc{feds}}
\DeclareMathOperator{\E}{\mathbb{E}}

\DeclarePairedDelimiter{\floor}{\lfloor}{\rfloor}
\newcommand{\leaves}{\textsc{leaves}}

\newcommand{\children}{\textsc{children}}
\newcommand{\descendants}{\textsc{descendants}}

\DeclarePairedDelimiter{\ceil}{\lceil}{\rceil}
\DeclarePairedDelimiter{\set}{\{}{\}}

\newcommand{\A}{\mathcal{A}}
\renewcommand{\vec}{\boldsymbol}

\newcommand{\sel}{{\textsc{sel}}}
\newcommand{\uns}{{\textsc{uns}}}
\newcommand{\topics}{\mathcal{T}}

\allowdisplaybreaks

\title{Federated Assemblies}

\author{%
  Daniel Halpern \\
  Harvard University\\
  \And
  Ariel D. Procaccia\\
  Harvard University\\
  \AND
  Ehud Shapiro\\
  Weizmann Institute of Science\\and London School of Economics\\
  \And
  Nimrod Talmon\\
  Ben-Gurion University\\
}

\begin{document}

\maketitle

\begin{abstract}
A \emph{citizens' assembly} is a group of people who are randomly selected to represent a larger population in a deliberation. While this approach has successfully strengthened democracy, it has certain limitations that suggest the need for assemblies to form and associate more organically. In response, we propose \emph{federated assemblies}, where assemblies are interconnected, and each parent assembly is selected from members of its child assemblies. The main technical challenge is to develop random selection algorithms that meet new representation constraints inherent in this hierarchical structure. We design and analyze several algorithms that provide different representation guarantees under various assumptions on the structure of the underlying graph. 
\end{abstract}

\section{Introduction}

\emph{Citizens' assemblies} are a popular mechanism for democratic decision making~\cite{Stone11,Rey16,Fish18,Land20}. In the last decade, this paradigm has vastly grown in recognition and influence. In Europe, for example, governments have sponsored citizens' assemblies to inform national policy on constitutional questions (Ireland), climate change (France), and even nutrition (Germany). Technology companies like Meta~\cite{Clegg23} are also piloting (enormous) citizens' assemblies as a way of obtaining democratic inputs for AI governance and alignment.

While different assemblies may take somewhat different approaches, they all share two distinctive features. First, members of a citizen's assembly are \emph{randomly selected} among volunteers. Second, members of the assembly engage in a long and substantial \emph{deliberation} before reaching any conclusions.

The former feature is of great technical interest, as it is challenging to design a good random selection process. The goal is to achieve \emph{descriptive representation}, in the sense that the assembly should reflect the composition of the population along multiple dimensions like gender, age, ethnicity and level of education; this is seen as a source of legitimacy for citizens' assemblies. However, since the pool of volunteers is typically skewed due to self-selection bias, uniform random selection will not yield descriptive representation and more sophisticated algorithmic solutions are required. Such algorithms, which are designed to achieve descriptive representation while optimizing fairness to volunteers, have been broadly deployed~\cite{Pro22}. 

\paragraph{Our proposal: Federated assemblies.}
To our knowledge, the hundreds of citizens' assemblies convened around the world have all been independent of each other: in collaboration with practitioners, different countries, regions and municipalities have organized their own assemblies from the ground up.

By contrast, we propose a novel form of citizens' assemblies: \emph{federated assemblies}. The most basic building block of a federated assembly is two assemblies (say, each representing the residents of a city) that decide to federate, forming a new parent assembly (which represents the residents of both cities and discusses policy questions of mutual interest); crucially, \emph{the members of the parent assembly are selected from the child assemblies}. More generally, a parent assembly can have more than two children, a child assembly can have more than one parent, and the parent assemblies themselves can federate. Overall, a federated assembly is represented by a directed acyclic graph, where nodes correspond to assemblies and an edge from $x$ to $y$ means that $x$ federated with other assemblies to form the parent assembly $y$. 
The lowest level (non-federated) assemblies are \emph{leaf assemblies}, which allow people to directly sign up as constituents. Even leaf assemblies need not represent only geographical entities, they can also correspond to issues (such as climate change) or identity groups (such as ethnic groups). Furthermore, people can sign up for multiple leaf assemblies\emdash any that intersect their multi-faceted interests.

In our view, federated assemblies have several advantages over the current practice of citizens' assemblies. First, the process of forming a new assembly does not have to start from the very beginning, as its members are selected from child assemblies; therefore, the lengthy and costly process of recruiting volunteers can be avoided altogether and the bar for forming an assembly is significantly lowered. Second, in standard citizens' assemblies, the determination of which features to stratify over, and which values to assign to these features\emdash which is made by the organizers\emdash is sometimes controversial and gives rise to manipulation opportunities. By contrast, in federated assemblies, these ``features'' \emdash which are induced 
by the structure of the graph \emdash are self-determined. Third, in the spirit of \emph{associative democracy}\emdash ``a model of democracy where power is highly decentralized and responsibility for civic well-being resides with like-minded civic associations''~\cite{JM10}\emdash federated assemblies allow citizens to exercise power by forming organizations that are immediately integrated into a broader framework of governance. We believe that federated assemblies may be especially pertinent in the context of a \emph{global citizens' assembly}\emdash the holy grail of practitioners of deliberative democracy\emdash as such an assembly could form organically as a federation of assemblies representing different countries, regions, and global issues.

\paragraph{Technical challenge and our results.} 

Our proposal is undoubtedly radical and we acknowledge that the devil is in the details; we discuss some limitations in Section~\ref{sec:disc}. Our goal in this paper is to address a key, technically challenging question that arises as we consider the implementation of federated assemblies: how should they be selected?

In the context of federated assemblies, we think of an assembly as satisfying descriptive representation if it reflects both its child assemblies and its constituents. Specifically, we wish to design a random selection process, where the members of each assembly are selected from its child assemblies, so that the following constraints are satisfied:
\begin{itemize}
	\item \emph{Individual Representation:} Let the assembly's \emph{population} be the union of all (possibly overlapping) populations of its descendant leaf assemblies. Each member of this population should have an equal probability of being represented on the assembly. This constraint can be interpreted as realizing an \emph{equality of power} ideal.
     \item \emph{Ex ante representation of child assemblies:} The expected number of seats allocated to each child assembly should be proportional to the child assembly's population.\footnote{When calculating a child's proportional share, if some populations overlap, we first split the weight of members in the intersection equally across their populations, e.g., if a member is in three child populations, they only contribute 1/3 to each.}
	\item \emph{Ex post representation of child assemblies:} The number of seats allocated to each child assembly should be proportional to the child assembly's population, rounded down, \emph{ex post}. This ex post guarantee prevents situations where an unlucky draw leads to significant under-representation; it mirrors ex post quotas imposed on different features in the selection of standard citizens' assemblies.  
\end{itemize}
Our results are primarily positive, showing that achieving various properties together is indeed tractable.

We begin by considering only the first two properties, individual representation and ex ante representation of child assemblies (\Cref{sec:exante}). We design a simple algorithm (\Cref{alg:priority-order}) that is able to achieve both of these properties under some minor regularity conditions.

Next, we throw ex post child representation into the mix (\Cref{sec:expost}). This makes the problem much more challenging, so we focus on instances with additional structure. We begin with a very natural class, which we call \emph{laminar instances}, where the assembly graph is a tree, and members are signed up for only one leaf assembly. This captures instances where assemblies represent hierarchical regions, e.g., city-level which feed into state-level which feed into national-level. For such instances, we give an algorithm (\Cref{alg:tree}) that achieves all of our desired properties. Next, we generalize laminar instances to a larger class, \emph{semi-laminar}. These are rooted in a laminar instance, with multiple assemblies at each node that federate together, allowing constituents to, for example, organize both geographically and according to shared interests. It turns out that even this level of generality already adds quite a bit of complexity to the problem. We nonetheless devise a surprisingly intricate algorithm (\Cref{alg:region-topic}), that again, under mild regularity conditions, is able to achieve individual representation, ex ante representation of child assemblies, and approximate ex post representation of child assemblies up to an additive error of one.

Finally, in \Cref{sec:experiments}, we implement an algorithm based on column generation for convex programming, which, given any instance of our problem, computes a distribution satisfying the three properties. In addition to measuring the running time of the algorithm, our empirical results suggest that all of our properties can be achieved simultaneously in the general case, at least in practice.

\paragraph{Related work.}
There is a growing body of work on algorithms for randomly selecting citizens' assemblies, starting with the paper by \citet{FGGH+21}; there is significant interest in this topic in AI conferences, especially NeurIPS~\cite{FGGP20,FKP21,EKMP+22,FLPW24}. The key challenge these papers address arises because of quotas imposed on multiple, overlapping features. By contrast, our community representation constraint amounts to stratified sampling with respect to a partition of the population, which is simple in the case of standard citizens' assemblies~\cite{BGP19}. The difficulty of our problem stems from its graph structure and the need to also represent child assemblies, leading to technical questions that are quite different from prior work. We do note that the mathematical programming approach used by \citet{FGGH+21} to implement their (widely used) algorithm also works well in our case. 

Conceptually, federated assemblies are somewhat related to \emph{pyramidal democracy}~\cite{Piv09}. In this scheme, citizens self-organize into small groups at the bottom of the pyramid, with each group nominating a delegate. In the next level, those delegates form groups, each again nominates a delegate, and so on. Our proposal differs in several significant ways, but most importantly, the selection of delegates is not random in pyramidal democracy\emdash rather, it is up to each group to decide how to select its delegate\emdash and there are no representation requirements; by contrast, the whole point of our work is to design a random selection process that satisfies representation constraints.

Our technical contribution fits more broadly within the literature on dependent rounding~\cite{gandhi2006dependent}. Arguably the closest in flavor is randomized rounding for flows~\cite{raghavan1987randomized}, which has constraints similar to our inheritance ones. Furthermore, rounding the number of seats given to each child assembly (a step in several of our algorithms) is reminiscent of randomized apportionment~\cite{balinski2010fair,pukelsheim2017proportional}. However, our inheritance and ex post child constraints do not fit neatly within existing frameworks and make it challenging to use off-the-shelf techniques directly. Instead, many of our results repeatedly invoke existing schemes, such as variations of the Birkhoff Von-Neumonn Theorem~\cite{Birk46,budish2013designing}, in nontrivial ways. We discuss these techniques in more detail when we use them.

\section{Model}

Let $N$ be a finite set of people and $G$ be a directed acyclic graph. Abusing notation slightly, we write $v \in G$ if $v$ is a node in $G$.

For a node $v \in G$, let
$\children(v) = \{v' :  v \rightarrow v'$ in $G\}$ be the set of nodes $v'$ with a directed edge from $v$ to $v'$. Let $\leaves(G) = \set{v:\ |\children(v)| = 0}$ be the leaf nodes of $G$. Each person $i \in N$ will be signed up for a nonempty set $L_i \subseteq \leaves(G)$ of leaf nodes. For a set $L \subseteq \leaves(G)$, we will also use $C^L = \set{i:\ L_i = L}$ for the set of people signed up for exactly the leaf nodes $L$. We refer to each $C^L$ as an \emph{equivalence class}, which collectively form a partition of $N$. For a leaf node $\ell \in G$, its \emph{population} $N_\ell = \set{i:\ \ell \in L_i}$ is the set of all people signed up for it.

Let $\feds(G) = \set{v:\ |\children(v)| > 0}$ be the internal nodes of $G$, which we refer to as \emph{federations}. For a federation $f \in \feds(G)$, let $\descendants(f) \subseteq \leaves(G)$ be the set of all leaf nodes reachable from $f$ in $G$. For a federation $f \in \feds(G)$, its population  is defined as $N_f := \bigcup_{\ell \in \descendants(f)} N_\ell$. Note that this can equivalently be defined in terms of equivalence classes as $N_f = \bigcup_{L: L \cap \descendants(f) \ne \emptyset} C^L$.

An \emph{instance} is a tuple $\mathcal{I} = \langle G, (N_v)_{v \in G} \rangle$ of the graph and the membership relationships.
Given an instance $\mathcal{I}$ and a \emph{target assembly size} $n$, our goal is to choose an \emph{assembly} $A_v \subseteq N_v$ of size $n$ for each node $v$. We will assume, in general, that each $|N_v| \ge n$ so that this is always possible. Furthermore, we require that each federation's assembly be drawn from its child assemblies, i.e., for $f \in \feds(G)$, $A_f \subseteq \bigcup_{c \in \children(f)} A_c$. We call this the \emph{inheritance} property. A vector $\vec{A} = (A_v)_{v \in G}$  satisfying these requirements an \emph{assembly assignment} (or simply an assignment). 
Our algorithm for selecting assembly assignments will be random, and thus, their outputs will be distributions over assembly assignments, $\vec{\mathcal{A}}$. We will call such a distribution a \emph{randomized assembly assignment} and use $\mathcal{A}_v$ to refer to the marginal distribution for the assembly at node $v$.

We would like our randomized assembly assignments to satisfy various properties. Some are ex post and should hold for all assignments in the support. Others will be ex ante and are simply properties of the distribution that hold in expectation.

\textbf{Desired Properties.}
Arguably, the most important requirement is \emph{individual representation}. A randomized assignment $\vec{\mathcal{A}}$ satisfies \emph{individual representation} if for each node $v$ and $i \in N_v$, $\Pr[i \in \mathcal{A}_v] = n/|N_f|$, that is, each person has an equal chance of being selected to the assembly. 

The other flavor of requirements we have on solutions are with respect to child assemblies. For a federation $f \in \feds(G)$ and child $c \in \children(f)$, we think of $|A_f \cap A_c|$ as the number of seats child $c$ is allocated, and we would like this allocation to be at least as large as $c$ ``deserves.'' The question, however, is how to set these bounds. If the child populations $N_{c}$ of a federation$f$ are all pairwise disjoint, then a natural choice is that each $c$ should be allocated an $|N_c|/|N_f|$ fraction of the $n$ seats in $A_f$. For non-disjoint child populations, we generalize this by splitting a person's weight equally among all child populations they are a part of. More formally, for a federation $f$ and member $i \in N_f$,  define the \emph{multiplicity of $i$ at $f$} to be $m(i,f) = |\{c \in \children(f) : i \in N_c\}|$, the number of child nodes they are a member of.
The \emph{weighted population size} $w_{c,f}$ of a child federation $c\in \children(f)$ is defined by  $w_{c,f} = \sum_{i\in N_c} \frac{1}{m(i,f)}$. 
Note that the definition implies that $|N_f| = \sum_{c\in \children(f)} w_{c,f}$.
Hence, we would say that node $c$ \emph{deserves} a $ w_{c, f}/|N_f|$ fraction of the seats in $A_f$.
Define $q_{c, f} := w_{c, f}/|N_f|$ to be this fraction.

However, $n \cdot q_{c, f}$ need not be an integer. Hence, we consider two notions of fair child representation. First, we say that a randomized assignment $\vec{\mathcal{A}}$ satisfies \emph{ex ante child representation} if, for each federation $f \in \feds(G)$ and child $c \in \children(f)$, $\E[|\mathcal{A}_f \cap \mathcal{A}_c|] \ge n \cdot q_{c, f}$. Similarly, we will say that a randomized assignment satisfies \emph{ex post child representation} if for each assignment $\vec{A}$ in the support and each federation $f \in \feds(G)$ and child $c \in \children(f)$, $|A_f \cap A_c| \ge \floor{n \cdot q_{c, f}}$. In other words, even if it is impossible to guarantee exact child representation, we can ensure that all children get at least the number of seats they deserve rounded down.

\section{Ex Ante Child Representation}\label{sec:exante}
We begin our study with an algorithm that samples from a distribution satisfying inheritance, individual representation, and ex ante child representation under mild conditions, proving that these three properties are compatible. To gain intuition, we first informally present a simpler version of the algorithm tailored to the special case of assemblies of size $n = 1$. It works by selecting a single random order of all of $N$ uniformly at random; the representative for node $v$ is simply the highest-ranked member of $N_v$. Note that this allows for strikingly simple arguments for the various properties. Indeed:
\begin{itemize}
    \item \emph{Inheritance}: If $A_f = \set{i}$ for a federation $f$ then $i \in N_c$ for some $c \in \children(f)$. Since $N_c \subseteq N_f$, $i$ will be maximal among $N_c$, and hence, $A_c = \set{i}$. 
    \item \emph{Individual Representation}: Each $i \in N_v$ is equally likely to be the maximally ranked person among $N_v$, so they are a member of $A_v$ with probability $1/|N_v|$.
    \item \emph{Ex ante child representation}: A similar argument to inheritance implies that $A_c = A_f$ with probability $|N_c| / |N_f| \ge w_{c, f}$. 
\end{itemize}

\begin{algorithm}[t]
        \caption{Selection algorithm with ex ante child representation guarantees}
        \label{alg:priority-order}
        \hspace*{\algorithmicindent} \textbf{Input} Graph $G$, populations $(N_v)_{v \in G}$, and assembly size $n$\\
        \hspace*{\algorithmicindent} \textbf{Output} Assembly assignment $(A_v)_{v \in G}$
        \begin{algorithmic}[1]
            \State Choose $n$ linear orders over $N$, $\succ^1, \ldots, \succ^n$ independently uniformly at random\;
            \State Select an additional linear order $\succ^{equiv}$ over $N$ uniformly at random.
            \For {$v \in G$}
                \State Let $i^{max}_{v, 1}, \ldots, i^{max}_{v,n}$ be the maximally ranked people in $\succ^1, \ldots, \succ^n$ when restricted to $N_v$
                \For{$L \subseteq \leaves(G)$}
                    \State $r^L_v \gets |\set{j:\ i^{max}_{j, v} \in C^L}|$, i.e., the number of $i^{max}_{j, v}$ members in equivalence class $C^L$
                    \State Let $B^L_v$ be the $r^L_v$ highest ranked members of $C^L$ in $\succ^{equiv}$
                \EndFor
                \State $A_v \gets \bigcup_L B^L_v$
            \EndFor
            \State \Return $(A_v)_{v \in G}$\;
        \end{algorithmic}
    \end{algorithm}

The challenge is to extend the foregoing algorithm to $n \ge 1$. One straightforward idea is to run the same algorithm $n$ independent times, which would produce $n$ singleton assemblies $A^1_v, \ldots, A^n_v$ for each node $v$, and subsequently set $A_v = \bigcup_j A^j_v$. At first glance, this seemingly maintains all of the properties of the $n = 1$ algorithm. However, this idea, unfortunately, does not quite compile; with some positive probability, we will have $j \ne j'$ with $A^j_v = A^{j'}_v$ so we do not end up with a size-$n$ assembly. We can, nevertheless, remedy this by replacing selected members with distinct people from their equivalence class. This fix does impose an additional requirement that each nonempty equivalence class contains at least $n$ people. However, this is a relatively mild condition since we view $n$ as a reasonably small constant and populations as potentially very large. Later, we discuss how even this mild assumption can be relaxed while only slightly degrading the guarantees.

\begin{theorem}\label{thm:priority}
    Assume that each nonempty equivalence class $C^L$ satisfies $|C^L| \ge n$. Then, \Cref{alg:priority-order} satisfies individual representation and ex ante child representation.
\end{theorem}
\begin{proof}
    Note that each $r_v^L \le n$ and $r_v^L > 0$ only if $C^L$ is nonempty. The requirement that each $C^L$ be of size at least $n$ ensures that line 7 is able to run, and we can always pick the $b_v^L$ highest ranked members. Without this, line 7 may fail.

    Since equivalence classes are disjoint and $\sum_L r_v^L = n$ (each $i^{max}_{v, j}$ will contribute to exactly one), we have that each $|A_v| = n$. Furthermore, if $i \in A_v$, there was another $i' \in N_v$ such that both $i$ and $i'$ are in the same equivalence class. Hence, $i \in N_v$ and the chosen assemblies are also valid. 

    We next show that each of the properties holds. For each $L \subseteq \leaves(G)$, node $v$, and $j \le n$, let $I^L_{v, j} = \mathbb{I}[i^{max}_{v, j} \in C^L]$ be the indicator variable that $i^{max}_{v, j}$ is signed up for the set of leaves $L$. As long as $i \in N_v$, we have that $\mathbb{E}[I^L_{v, j}] = |C^L| / |N_v|$ (over the randomness of the selected orders).

    We begin with individual representation. Fix a node $v$ and a person $i \in N_v$. Suppose we condition on a specific value of $r_v^L$, then (abusing notation slightly) over the randomness of $\succ^{equiv}$, we have that $\Pr[i \in \mathcal{A}_v:\ r_v^L] = r_v^L / |C^L|$ because each person $i \in C^L$ is equally likely to be in any of the $|C^L|$ positions. Next, note that $r_v^L = \sum_j I^L_{v, j}$. Hence, $\E[r_v^L] = n|C^L| / |N_v|$. Putting these together, we have that $\Pr[i \in \mathcal{A}_v] = n / |N_v|$, as needed.

    Next, we show inheritance. Fix a federation $f \in \feds(G)$ and an equivalence class $C^L$ such that $C^L \subseteq N_f$. Note that there must be some child $c \in \children(f)$ such that $C^L \subseteq N_f$. We will show that for this choice of $c$, $B^L_f \subseteq B^L_c \subseteq A_c$, ex post. As this holds for every $L$, it follows that $A_f = \bigcup_L B^L_f \subseteq \bigcup_{c \in \children(f)} A_c$, ex post. To that end, it is sufficient to show that $r^L_f \le r^L_c$, which implies that $B^L_f \subseteq B^L_c$. For this, we can simply show that for each $j$, $I^L_{v, f} \le I^L_{v, c}$, ex post. Indeed, if the maximal selected member of $\succ_j$ when restricted to $N_v$ is a member of $C^L$, then the same person will be maximal when restricted to $N_c$ because $C^L \subseteq N_c$.

    Finally, we show ex ante child representation. The key observation is that for a child $c \in \children(f)$, $|A_c \cap A_f| = \sum_{L : C^L \subseteq N_c} r^L_f$ because, as shown above, $r^L_f \le r^L_c$, so the top $r^L_f$ people from $C^L$ are contained in both $A_c$ and $A_f$. It follows that
    \[
        \E[\mathcal{A}_c \cap \mathcal{A}_f] = \sum_{L : C^L \subseteq N_c} \E[r^L_f] = \sum_{L : C^L \subseteq N_c} \frac{|C^L| \cdot n}{|N_f|} =\frac{|N^c| \cdot n}{|N_f|} \ge q_{c, f} \cdot n,
    \]
    as needed.
\end{proof}

We now discuss ways to implement a modification of \Cref{alg:priority-order} that works even when $|C^L| < n$. The key idea is that \Cref{alg:priority-order} will only fail to run if there is a node $v$ such that $r^L_v > |C^L|$. In such cases, we can simply ``reject'' and restart the algorithm. Note that ex post guarantees will still be satisfied as long as the algorithm is able to terminate. For ex ante guarantees, as long as the probability of failure is at most some value $p$, the properties only degrade as a function of $p$. Specifically, each $i \in N_v$ will be selected in $A_v$ with probability at least $n/|N_v| - p$, and for a child $c \in \children(f)$, their expected intersection will be at least $n (w_{c,f} - p)$. As long as $p \ll 1/|N_v|$, this will be a negligible loss. We show that, under mild conditions on the populations being reasonably large and equivalence classes being at least of size $3$, this is indeed the case; see \Cref{app:degraded-guarantees} for more details.

\section{Ex Post Child Representation}\label{sec:expost}

In this section, we add ex post child representation to our list of requirements, albeit at a cost to the generality of our results.  

\subsection{Laminar Instances}

We begin with a more restricted structure that captures many practical potential implementations of federated assemblies. The assumption is that $G$ is a tree and that each $i \in N$ is signed up for exactly one leaf node, i.e., $|L_i| = 1$. This captures settings such as where the assemblies represent regions that form a hierarchy, i.e., city-level assemblies, which feed into state-level assemblies, which feed into national-level assemblies, and possibly beyond. Following set-theoretic terminology, we call instances satisfying these restrictions \emph{laminar}.

\Cref{alg:tree} works by going through the federations,  allocating an integral number of seats to each child, and filling these seats directly from the child's assembly. The rounding (Line 4) can be done in a variety of ways using tools from the dependent rounding literature. \citet{brewer1983introduction} provide a number of classical statistical methods for this; canonical examples from randomized algorithm design include \emph{pipage rounding}~\cite{ageev2004pipage} and variations of the \emph{Birkhoff-von Neumann Theorem}~\cite{Birk46,budish2013designing}. 
\begin{algorithm}[t]
    \caption{Selection algorithm for laminar instances with ex post child representation guarantees}
    \label{alg:tree}
    \hspace*{\algorithmicindent} \textbf{Input} Graph $G$, populations $(N_v)_{v \in G}$, and assembly size $n$\\
        \hspace*{\algorithmicindent} \textbf{Output} Assembly assignment $(A_v)_{v \in G}$
    \begin{algorithmic}[1]
        \For{leaf $\ell \in \leaves(G)$}
            \State Choose $A_\ell \subseteq N_\ell$ uniformly at random
        \EndFor
        \For{federation $f \in \feds(G)$ in a topological sort}
            \State Round $(n \cdot q_{c, f})_{c \in \children(f)}$ to an integral vector $(s_{c, f})_{c \in \children(f)}$ such that 
            \Statex\hspace{0.44in}  each $s_{c, f} \ge \floor{w_{c, f}}$, $\mathbb{E}[s_c] = n \cdot w_{c, f}$, and $\sum_c s_{c, f} = n$
            \State Select $B_{c, f} \subseteq A_c$ with $|B_{c, f}| = s_{c, f}$ uniformly at random.
            \State Let $A_f \gets \bigcup_c B_{c, f}$. 
        \EndFor
        \State \Return $(A_v)_{v \in G}$\;
    \end{algorithmic}
\end{algorithm}

\begin{theorem}
\label{thm:tree}
    \Cref{alg:tree} satisfies individual representation, ex ante child representation, and ex post child representation on laminar instances.
\end{theorem}

While the proof is relegated to \Cref{app:tree}, we give some intuition by discussing the naturalness of \Cref{alg:tree}, which allows for its relatively simple analysis. Specifically, we enforce ex ante and ex post child representation by first allocating the ``correct'' number of seats to each child. We then go through the tree iteratively, selecting members from their (already determined) child assemblies. We may hope that such ideas could be generalized to non-laminar instances, first allocating seats and then iteratively selecting members from the child assemblies, perhaps not uniformly at random, to account for people being in potentially different numbers of children. However, we give an example where this is not the case \emdash relaxing either of the laminar assumptions (that $G$ is a tree or that each $|L_i| = 1$) can lead to instances where satisfying all the properties is impossible using such algorithms. Instead, it is necessary to induce some other forms of correlation or relax the iterativeness, a challenge that leads to more complicated algorithms. We discuss this more formally in \Cref{app:iterative}.

\subsection{Semi-Laminar Instances}\label{subsec:semi-laminar}


We now turn to a generalization of laminar instances with a structure we view as quite practical for real-world implementation, as it allows people to organize both geographically and according to shared interests. Conceptually, there is an underlying laminar instance with graph $R$. In addition, there is a set $\topics$ representing a set of \emph{topics}. These may be various causes that people care about, say climate change or animal rights, or region-specific policy questions, such as budget allocation. 

More specifically, the graph $G$ has $|R| \cdot (|\topics| + 1)$ nodes. $|R| \cdot |\topics|$ of these are identified by members of $R \times \topics$, i.e., there is a node $(r, t)$ for each combination of region and topic.  For each $t \in \topics$, the set of nodes $R \times \set{t}$ is connected to form a copy of $R$. In addition, there is a set of nodes denoted $(r, *)$ for each $r \in R$. Each is a federation whose children are $\set{r} \times \topics$. In other words, at each region, we have an assembly that represents all people with respect to all topics of that region. 

In this graph, the leaves are nodes in $\leaves(R) \times \topics$. People can be signed up for any number of topics, but, as the underlying instance is laminar, we assume that each is a member of nodes in one region. Hence, we assume that each $L_i \subseteq \set{r} \times \topics$ for some $r \in \leaves(R)$.

Abusing notation slightly, we will write $N_r = \set{i \mid L_i \subseteq \set{r} \times \topics}$ for all people in this region (technically $N_r = N_{(r, *)}$). Furthermore, a node $(r, t)$ will have at most two parents in $G$: it will always have $(r, *)$ and possibly another node $(r', t)$ if $r$ has a parent ($r'$) in $R$. Note that if $(r', t)$ exists, the weighted population $w_{(r, t), (r', t)} = N_{(r, t)}$, the trivial weighting, because the children of $(r', t)$ have disjoint populations. On the other hand, $w_{(r, t), (r, *)} = \sum_{i \in N_{(r, t)}} \frac{1}{|L_i|}$. For brevity, we will use $w_{r, t}$ to denote only the nontrivial weight of the node $(r, t)$. Finally, note that, for $r \in \leaves(R)$ and $T \subseteq \topics$, we can write $C^{\set{r} \times T}$ for the equivalence class of people in leaf $r$ signed up for topics $T$. Everybody must fall in exactly one of these equivalence classes.


\begin{algorithm}[t]
    \caption{Algorithm for semi-laminar instances}
    \label{alg:region-topic}
    \hspace*{\algorithmicindent}\textbf{Input} Semi-laminar instance with graph $G$, populations $(N_v)_{v \in G}$, and assembly size $n$ \\
    \hspace*{\algorithmicindent}\textbf{Output} Assembly assignment $(A_v)_{v \in G}$
        
    \begin{algorithmic}[1]
        \For {$(r, t) \in R \times \topics$}
            
            \State  $s_{r, t} \gets$ \Call{Round}{$n \cdot \frac{w_{r, t}}{|N_{r, t}|}$}  \; \label{lin:srt-rounding}
            
        \EndFor
        \For{$r \in \leaves(R)$}
            \State $(B^\sel_{r, t}, B^\uns_{r, t})_{t \in T} \gets$ \Call{SampleLeaves}{$(C^{\set{r} \times T})_{T \subseteq \topics}$, $(s_{r, t})_{t \in \topics}$, $n$} \;
        \EndFor
        \For{$r \in \feds(R)$ in a topological sort}
            \For{$t \in \topics$}
                \State $B^\sel_{r, t}, B^\sel_{r, t} \gets $ \Call{SampleFromChildren}{$s_{r, t}$, $(B^\sel_{c, t}, B^\uns_{c, t}, w_{c, t}, |N_{c, t}|)_{c \in \children(r)}$}
            \EndFor
        \EndFor
        \For{$r \in R$}
            \For{$t \in \topics$}
                \State $A_{(r, t)} \gets B^\sel_{r, t} \cup B^\uns_{r, t}$\;
            \EndFor
            \State $A_{(r, *)} \gets$ \Call{RoundAndSample}{$n$, ($B^\sel_{r, t}, w_{r, t})_{t \in \topics}$}\label{lin:aggregateround}
        \EndFor
        \State \Return $(A_v)_{v \in G}$\;
    \end{algorithmic}
\end{algorithm}

We refer to an instance taking on the above structure as a \emph{semi-laminar instance}. For such instances, we have the following algorithm shown in \Cref{alg:region-topic}, with additional helper functions formally defined in \Cref{alg:helpers} in \Cref{app:helpers}. The structure is essentially an extension of the algorithm for laminar instances. Ideally, for each topic, we could independently run \Cref{alg:tree}, and when we needed to select an assembly $(r, *)$, we could select the ``correct'' number of members from each $A_{(r, t)}$. However, this does not quite work because if people are signed up for more topics, this will lead to them ending up in $A_{(r, *)}$ more frequently. To account for this, we essentially partition each $A_{(r, t)}$ into two pieces, $B^\sel_{r, t}$ of \emph{(sel)ectable} people and $B^\uns_{(r, t)}$ of \emph{(uns)electable} people, and run a separate laminar-like algorithm for each of these. The key idea is that when selecting members of $A_{(r, t)}$ for $A_{(r, *)}$, we will only select from members in $B^\sel_{r, t}$. While each $i \in N_{(r, t)}$ will have an equal chance of being in $A_{(r, t)}$, the more topics they are signed up for, the less likely they will be in $B^\sel_{r, t}$, so that, in aggregate, this will not increase their chances of being in $A_{(r, *)}$.

It should be said that the real complexity of this algorithm is hidden in the dependent rounding schemes of the helper functions. It requires careful balance and specific constraints to ensure the probabilities are exact and not subject to strange correlations that na\"ive implementations can impart. Furthermore, we need to avoid scenarios where the same person is selected twice for $A_{(r, *)}$ from two separate topics, which would not appear to have an easy solution. All of these schemes are implemented using variations of the Birkhoff-von Neumann algorithm. \citet{budish2013designing} give a class of constraints that are always possible to guarantee when rounding; we ensure that all of our rounding procedures take this form.

Finally, note that all of these extra complexities mean that we impose some additional mild regularity conditions and achieve slightly degraded guarantees, at least for ex post child representation. The regularity conditions require that no weighted population is too close to zero or the entire population and that no child population is too dominant within its parent.

\begin{theorem}\label{thm:topic}
    Fix a semi-laminar instance and assembly size $n$. Suppose there exist $\varepsilon, \delta > 0$ such that (i) for all $r, t \in R \times \topics$, $\varepsilon \cdot |N_{r, t}| \le w_{r, t} \le (1 - \varepsilon) \cdot |N_{r, t}|$, (ii) for all federations $f \in \feds(G)$ and $c \in \children(f)$, $|N_c|/|N_f| \le 1 - \delta$, and (iii) $n \ge \frac{2}{\varepsilon \cdot \delta}$. Furthermore, suppose each $|N_v| \ge 4n$ and for all nonempty equivalence classes $C^L$, $|C^L| \ge 2$. Then, \Cref{alg:region-topic} run on this instance satisfies individual representation, ex ante child representation, and approximate ex post child representation in that $|A_f \cap A_c| \ge \floor{n \cdot q_{c, f}} - 1$ for all $f \in \feds(G)$ and $c \in \children(f)$.
\end{theorem}
The proof is relegated to \Cref{app:region}.

\section{Experiments}\label{sec:experiments}
While we have given efficient algorithms for finding randomized solutions satisfying various properties in special cases, we focus now on how easy it is in practice to find randomized assignments in generality. To this end, we sample several thousand instances and use a brute force algorithm to try to find a randomized assignment satisfying all of our ex post and ex ante guarantees.

\textbf{Algorithm.}
The computation of randomized assignments satisfying our guarantees is quite challenging. Our algorithm for this task uses convex optimization and integer linear programming (ILP) as subroutines. It is inspired by a related algorithm used for standalone citizens' assemblies, which is an extension of column generation to convex programming~\cite{FGGH+21}. At a high level, we convert our problem to a smoother optimization problem by defining a convex loss over randomized assignments measuring how far a distribution is from satisfying all ex ante guarantees.
The algorithm maintains a set of (deterministic) assembly assignments that all satisfy ex post guarantees. It alternates between two steps. First, it finds the distribution over just this set of assignments that minimizes the squared error. If this distribution achieves zero error (or, more accurately, some small additive error of .1\% to avoid numerical issues), then every ex ante guarantee is satisfied; in that case, we have found a solution and subsequently return it. On the other hand, if it does not, the algorithm runs an ILP to find an assignment satisfying ex post guarantees that maximizes the loss gradient the current best possible point. We iterate these two steps until a solution is found.

\textbf{Experimental Setup.}
We draw instances as follows. First, we fix a number of equivalence classes (2, 5, 10, or 20) and then sample their relative sizes from an exponential distribution. Next, we iterate through a number of federations (2, 5, 10, or 20). Each federation has a randomly selected set of already defined federations or equivalence classes as children.\footnote{This means that we allow federations to have direct members and child assemblies. However, this generality only makes the problem harder.} This method allows for arbitrary DAGs to be sampled. We then test these instances with various assembly sizes $n$ (2, 5, 10, or 20). 

For each combination of parameters, we sampled and ran 100 instances.
All optimizations were solved using Gurobi on an Amazon Web Services (AWS) instance with 128 vCPUs of a  3rd Gen AMD EPYC running at 3.6GHz equipped with 1TB of RAM. We were able to run 64 instances in parallel, giving each thread two processors. Depending on the size of the instance, computation took anywhere from under a tenth of a second to multiple hours.

\begin{figure}
    \centering
    \begin{subfigure}[b]{0.32\textwidth}
        \centering
        \includegraphics[width=\textwidth]{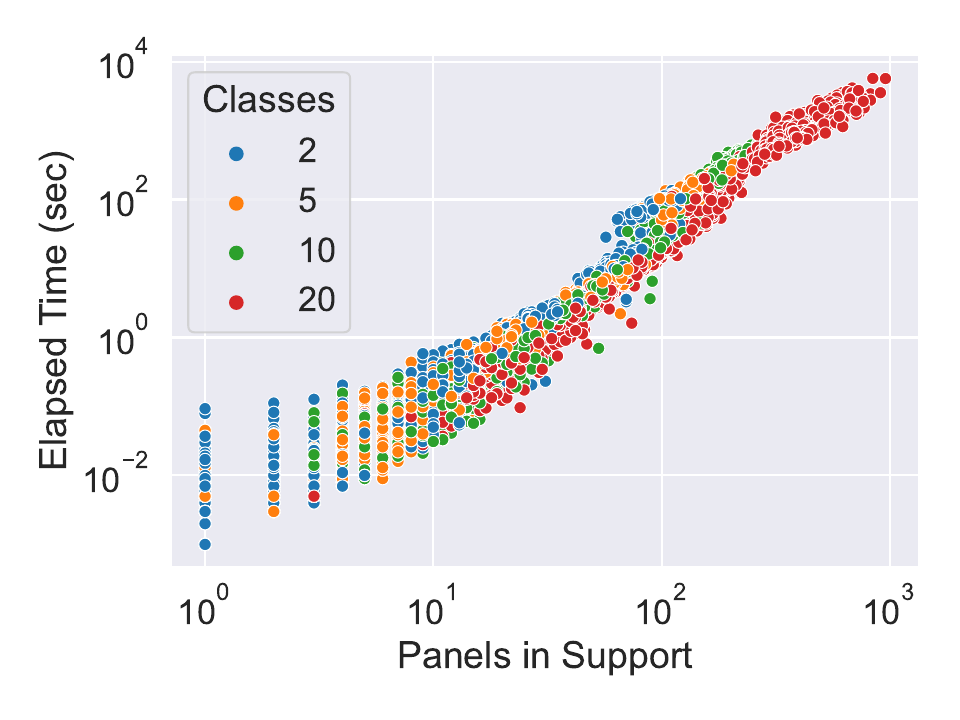}
        
        \caption{Instances colored by the number of equivalence classes.}
    \end{subfigure}%
    \hfill
     \begin{subfigure}[b]{0.32\textwidth}
        \centering
        \includegraphics[width=\textwidth]{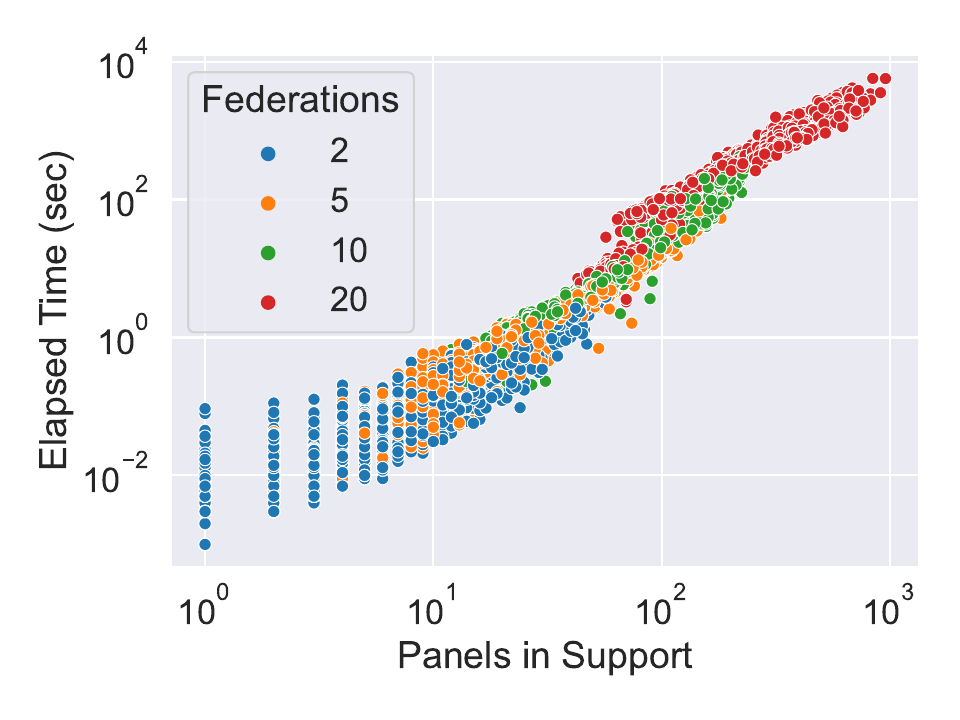}
        \caption{Instances colored by the number of federations.}
    \end{subfigure}%
    \hfill
     \begin{subfigure}[b]{0.32\textwidth}
        \centering
        \includegraphics[width=\textwidth]{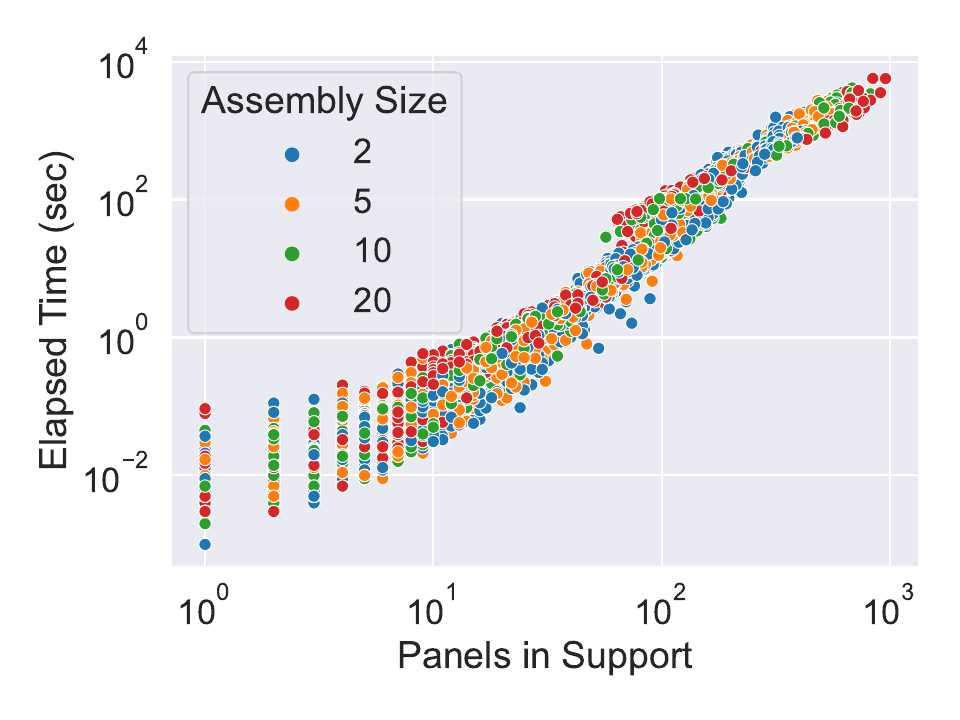}
        \caption{Instances colored by the assembly size.}
    \end{subfigure}
    \caption{Scatter plots showing the time taken and number of panels in the support for each of the instances we ran on. Sub-plots show the same plot, colored by a parameter.}
    \label{fig:plots}
\end{figure}

\textbf{Results.}
Of the 6400 instances, on all but 15, the algorithm terminated, returning an optimal solution. These 15 had only the largest number of equivalence classes and federations (20 each). With better optimization, in practice, the problem of finding a randomized assignment satisfying all of our guarantees appears to be both feasible and tractable.  

For the >99.7\% of instances that did terminate, 
we use two metrics to measure the complexity of finding the distributions: the elapsed time and the number of panels added to the support. A scatter plot showing the relationship between these parameters is given in \Cref{fig:plots}. Each subfigure shows the same points but colors them depending on a different parameter so that we may see its effect. Overall, we see that both of these complexity measures are quite similar. Furthermore, increasing the number of equivalence classes and federations strongly correlates with increased complexity. Assembly size, on the other hand, appears to be relatively unimportant.

\section{Discussion}
\label{sec:disc}

The democratic governance of large-scale digital communities is an open problem.  Key challenges include,  first, the penetration of fake and duplicate digital identities (a.k.a.~sybils), and
second, the perils of large-scale online voting, which is considered to be untenable by some leading experts~\cite{park2021going}. 
Federated assemblies can be viewed as a step in an effort to address these challenges.  
Our approach may address sybils by having the laminar core built from small local communities in which members know each other to be genuine and from communities that federate only if they trust each other to be genuine,  for example by having sufficient intersection or actual relationships among them to base this trust on.
Large-scale online voting is a nonissue as every federation, no matter how large, is governed by an assembly that engages in ``small-scale'' democracy.

Our approach also has some limitations.  First and foremost, an obvious barrier to federated assemblies is the question of who would set up such assemblies, manage the infrastructure, and provide funding? One way to address this challenge is to use the conceptual framework and architecture of grassroots systems~\cite{shapiro2023grassroots,shapiro2023gsn} and to construct the application of grassroots federated assemblies as a grassroots platform~\cite{shapiro2024grassroots}, operated on peoples' smartphones without relying on any global resources other than the network itself.

Second, we modeled the problem as a static and single-shot: we simply needed to sample a single assignment fairly. In practice, however, these assemblies are dynamic and must be periodically updated. This is not an inherent limitation, however. Indeed, one solution is to resample fresh assemblies every fixed amount of time. However, this may get more challenging if the system is more malleable with members coming and going. We may hope for more ``ex-post over time'' properties to ensure no group is consistently receiving the short end of the stick. Furthermore, we may hope to allow for local changes to occur without completely refreshing all assemblies simultaneously, say rotating people in and out one at a time, and doing so with minimal changes to assemblies on the opposite side of the graph. Modeling this well and defining useful ``over-time'' fairness properties seems to be challenging yet potentially impactful future work.

Finally, although we analytically solve well-motivated special cases, we leave open whether a randomized assignment satisfying all of our desiderata exists in the general case. In our extensive experiments, we have not found any infeasible instances, and we are therefore optimistic that existence can be guaranteed.

\subsection*{Acknowledgements}
Ehud Shapiro is the Incumbent of The Harry Weinrebe Professorial Chair of Computer Science and Biology at the Weizmann Institute and a Visiting Professor at the London School of Economics.

\bibliographystyle{plainnat}
\bibliography{abb,ultimate,bib}

\newpage

\appendix

\section{Extending \Cref{alg:priority-order} to Smaller Equivalence Classes}\label{app:degraded-guarantees}

We will analyze here the probability of \Cref{alg:priority-order} needing more than $|C^L|$ people from each equivalence class $C^L$. It is sufficient to ensure this does not occur in the leaves, as in the proof of \Cref{thm:priority} internal nodes select fewer people from $|C^L|$ than their children.

Fix an equivalence class $C^L$ and a leaf node $\ell \in L$. We can directly analyze the probability that more than $|C^L|$ people are selected from $C^L$ for this assembly. Namely, for each of the $n$ draws, the probability that it was a member from $C^L$ is $|C^L|/|N_{\ell}|$. We are doing $n$ draws of this. Hence, we wish to analyze the probability of $\Pr[\sum_{i=1}^n X_i > |C^L|]$, where each $X_i$ is drawn from an independent Bernoulli with $\Pr[X_i = 1] = |C^L|/|N_{\ell}|$. By Chernoff bound says that for all $X$ which is the sum of independent variables with $\mu = \E[X]$ and $\delta \ge 0$,
\[
    \Pr[X > (1 + \delta)\mu] < \left(\frac{e^\delta}{(1 + \delta)^{(1 + \delta)}} \right)^\mu \le \left(\frac{e}{1 + \delta} \right)^{(1 + \delta)\mu}.
\]
For our purposes, $\mu = n|C^L|/|N_\ell|$, and we wish to set $\delta$ such that $\mu(1 + \delta) = |C^L|$, so $1 + \delta = N^|ell|/ n$. This implies that the probability of failure is at most
\[
    \left(\frac{e \cdot n}{|N_\ell|}\right)^{|C^L|}.
\]
Suppose additionally that populations are not too small in the sense that $|N_\ell| \ge q |N|$ for some value $q$. Then, this is at most
\[
    \left(\frac{e \cdot n}{q \cdot |N|}\right)^{|C^L|}.
\]
Finally, to ensure this happens for no leaves or equivalence classes, we need to union bound over the at most $|G|$ leaf nodes and $|N|$ equivalence classes, leading to a total probability of failure of at most
\[
    |N| |G| \cdot \left(\frac{e \cdot n}{q|N|}\right)^{|C^\ell|} .
\]

Now, suppose each $|C^\ell|$ is of size at most $3$. Furthermore, assume that $|N| \gg n, 1/q, |G|$, which we would expect for reasonable instances. Thus, it makes sense to interpret this asymptotically as $|N|$ grows large compared to the rest of the terms. This leads to a bound of failure of $O(1/|N|^2)$, essentially negligible compared to all ex-ante guarantees (which are $\Theta(1/|N|)$).

\section{Proof of \Cref{thm:tree}}
\label{app:tree}

    Showing inheritance, ex ante child representation, and ex post child representation follow immediately from the definition of the algorithm.

    It remains to establish individual representation. Fix a person $i$, and let $\ell \in \leaves(G)$ be the leaf node they are signed up for. Note that $A_\ell$ is simply a random sample of $n$ people from $N_\ell$, so clearly $\Pr[i \in \mathcal{A}_\ell] = n/|N_\ell|$ . Next, fix a federation $f^* \in \feds(G)$ such that $i \in N_{f^*}$. Consider running the algorithm in a different order, sampling all vectors $(s_{c, f})_{c \in \children{f}}$ at the beginning before starting the algorithm, and then running according to these samples. Note that this leads to an equivalent process because each $(s_{c, f})_{c \in \children{f}}$ is sampled independently of everything else in the algorithm. Condition on a specific sample of these vectors. Let $\ell = v_0, v_1, \ldots, v_k = f$ be the path in $G$ that leads from $\ell$ to $f$. Note that we can now directly compute the probability $i \in A_{f^*}$ because the only way to do so is if $i \in A_\ell$ and $i \in B_{v_{j - 1}, v_j}$ for each $j \ge 1$. Hence, this probability is exactly
    \[
        \frac{n}{|N_\ell|} \cdot \prod_{j = 1}^k \frac{s^{v_{j - 1}, v_j}}{n}.
    \]
    To get the unconditional probability, we can simply take the expectation over all $s$ values, i.e.,
    \[
        \E\left[\frac{n}{|N_\ell|} \cdot \prod_{j = 1}^k \frac{s^{v_j}_{v_{j - 1}}}{n}\right].
    \]
    Since each $s_{c, f}$ and $s_{c', f'}$ is sampled independently for $f \ne f'$, we can push the expectation in to get equality to
    \[
        \frac{n}{|N_\ell|} \cdot \prod_{j = 1}^k \frac{ \E\left[s^{v_{j - 1}, v_j}\right]}{n} = \frac{n}{|N_\ell|} \cdot \prod_{j = 1}^k \frac{ n \cdot q_{v_{j - 1}, v_j}}{n} = \frac{n}{|N_\ell|} \prod_{j = 1}^k \frac{|N_{v_{j - 1}}|}{|N_{v_j}|} = \frac{n}{|N_{f^*}|},
    \]
    as needed. \qed

\section{Impossibility for Iterative Algorithms}\label{app:iterative}
To formalize this impossibility result, we call an algorithm \emph{topologically iterative} if it has the same structure as \Cref{alg:tree}, except line 2 and 5 are replaced with potentially different sampling schemes.

\begin{proposition}
    For all assembly sizes $n$, there exist instances both with $G$ being a tree (but people signed up for multiple leaves) or people signed up for a single leaf (but $G$ not being a tree) where no topologically iterative algorithm can simultaneously satisfy individual representation and ex ante child representation. 
\end{proposition}
\begin{proof}
    We begin with an instance where $G$ is a tree, but some people are signed up for multiple leaves. Fix an assembly size $n$. There will be a single federation $f$ with $2n$ children $c_1, \ldots, c_{2n}$ which are all leaves.  We will define populations by equivalence classes. There will be $2n$ sets of people $C^{\set{c_j}}$ of equal size that each are only signed up for $c_j$, i.e., $|C^{\set{c_j}}| = k$ for some arbitrary integer $k$. In addition, there will be a set $C^{\set{c_1, \ldots, c_{2n}}}$ of people signed up for all children assemblies. This will be $2n-1$ times as large as each individual group, so $|C^{\set{c_1, \ldots, c_{2n}}}| = (2n - 1)k$. Note that, by symmetry, $q_{c_j, f} = 1/(2n)$ for each $j$.
    
    Suppose our strategy now is to (1) sample $(A_{c_1}, \ldots, A_{c_{2n}})$ (from some distribution satisfying individual representation) (2) independently, sample an integral vector $(s_1, \ldots, s_{2n})$ such that $\sum_{j} s_j = n$ and $\E[s_j] = n/(2n) = 1/2$, and then (3) choose $A_f$ by selecting $s_j$ people from $A_{c_j}$. We now claim that we cannot select $A_f$ such that it satisfies individual representation. Indeed fix an arbitrary $i \in C^{\set{c_1}}$. A necessary condition for $i \in A_f$ is that both $i \in A_{c_1}$ and $s_1 \ge 1$.
    Note that $\Pr[i \in \A_{c_1}] = \frac{n}{|C^{c_1}| + |C^{\set{c_1, \ldots, c_{2n}}}|} = \frac{1}{2k}$. Furthermore, $\Pr[s_1 \ge 1] \le 1/2$ by Markov's inequality. Since these are selected independently, the probability they both occur is at most $\frac{1}{4k}$. However, individual representation ensures that $\Pr[i \in \A_f] = \frac{n}{|C^{\set{c_1, \ldots, c_{2n}}}| + \sum_{j} |C^{\set{c_j}}|} = \frac{n}{(4n - 1)k} > \frac{1}{4k}$.

    To convert this to an instance where $G$ is not a tree, we can have an additional node for each equivalence class, and have these nodes point to the set of leaf assemblies that class was signed up for. Then, each person can sign up for this single corresponding leaf node rather than a larger set of nodes, and the same argument goes through.
\end{proof}

\section{\Cref{alg:region-topic} Helper Functions}\label{app:helpers}

Here, we formalize the helper functions used in \Cref{alg:region-topic}, presented as \Cref{alg:helpers}. All of the randomized rounding can be done using a variation of the Birkhoff Von-Neumonn Theorem. Namely, \citet{budish2013designing} give an algorithm to handle the following randomized rounding instances. Say we are given a set of values $p_1, \ldots, p_k$ and a set of \emph{constraints} represented as a family of sets $\mathcal{I}$ where each $I \subseteq \mathcal{I}$ has $I \subseteq \set{1, \ldots, k}$. Our goal is to round these values to $x_1, \ldots, x_k$ such that $\mathbb{E}[x_i] = p_i$. Furthermore, we will do this such that for all $I \in \mathcal{I}$, $\floor{\sum_{i \in I} p_i} \le \sum_{i \in I} x_i \le \ceil{\sum_{i \in I} p_i}$. \citet{budish2013designing} show that if $\mathcal{I}$ is a \emph{bihierarchy}, then this is possible and can be done with a polynomial time algorithm. $\mathcal{I}$ is said to be a bihierarchy of there exists a partition $\mathcal{I} = \mathcal{I}_1 \cup \mathcal{I}_2$ such that for all pairs $I, I' \in \mathcal{I}_j$ of either partition, either $I \subseteq I'$, $I' \subseteq I$, or $I \cap I' = \emptyset$. One can check that all randomized roundings we do in the helper functions take this form.
\begin{algorithm}[t]
\caption{Helper functions}
    \label{alg:helpers}
    \begin{algorithmic}[1]
        \Function{Round}{$x$} 
            \State \Return $\ceil{x}$ with probability $x - \floor{x}$ and $\floor{x}$ otherwise.
        \EndFunction

        \Function{SampleLeaves}{$(C^T)_{T \subseteq \topics}$, $(s_t)_{t \in \topics}$, $n$}
        \State $p^T_t \gets \frac{|C^T|/|T|}{\sum_{T': t \in T'} |C^{T'}| / |T'|} \cdot s_t$ for all $T \subseteq \topics$ with nonempty $C^T$ and $t \in T$.
        \State Round all $(p^T_t)_{t, T}$ to $(\alpha^T_t)_{t, T}$ such that $\mathbb{E}[\alpha^T_t] = p^T_t$, $\alpha^T_t \le \ceil{ p^T_t}$,$\forall t, \sum_{T: t \in T} \alpha^T_t = s_t$,
        \Statex \hspace{0.44in}   and $\forall T, \sum_{t \in T} \alpha^T_t \le \ceil{\sum_{t \in T} p^T_t}$.
        \For{$T \subseteq \topics$}
            \State Choose $(D^T_t)_{t \in T}$ of sizes $(\alpha^T_t)_{t \in T}$ randomly from $C^T$ such they are each disjoint \label{lin:sample-D}
        \EndFor
        \For{$t \in \topics$}
           \State $B^\sel_t \gets \bigcup_{T: t \in T} D^T_t$ 
        \EndFor
        \For{$t \in \topics$}
        \State $q^T_t \gets \frac{|C^T|(1 - 1/|T|)}{\sum_{T': t \in T'} |C^{T'}|(1 - 1/|T'|)} \cdot (n - s_t)$ for all $T \subseteq \topics$ with nonempty $C^T$ and $t \in T$
            \State Round $(q^T_t)_{T}$ to $(\beta^T_t)$ such that each $\beta^T_t \le \ceil{q^T_t}$, $\sum_T \beta^T_t = n - s_t$, and $\E[\beta^T_t] = q^T_t$. 
            \For{$T: t \in T$}
                \State Choose $E^T_t$ be a random sample of $\beta^T_t$ members of $C^T \setminus \set{D^T_t}$.\label{lin:sample-E}
                \EndFor
                \State $B^\uns_t \gets \bigcup_{T: t \in T} E^T_t$.
        \EndFor
        \State \Return $(B^\sel_t, B^\uns_t)_{t \in \topics}$
        \EndFunction
        
        \Function{SampleFromChildren}{$s$, $n$, $(B^\sel_{c}, B^\uns_{c}, w_{c}, |N_{c}|)_c$}
            \State Let $x^\sel_c = s \cdot \frac{w_c}{\sum_{c'} w_{c'}}$ and $x^\uns_c = (n - s) \frac{|N_c| - w_c}{\sum_{c'} |N_{c'}| - w_{c'}}$ for all $c$.
            \State Round each $x^j_c$ for $j \in \set{\sel, \uns}$ and child $c$ to $\gamma^j_c$ such that $\E[\gamma^\sel_c] =  x^\sel_c$, $\gamma^{\sel}_c \ge \floor{ x^\sel_c}$, \label{lin:massiverounding}
             \Statex\hspace{0.44in} for $\sum_{c} \gamma^{\sel}_c = s$, $\sum_c \gamma^{\uns}_c = n - s$, and for all $c$, $\gamma^{\sel}_c + \gamma^{\uns}_c \ge \floor{x^\sel_c + x^\uns_c}$.
             \For{all $c$}
                \State Let $D^\sel_c$ be a random sample of size $\gamma^\sel_c$ from $B^\sel_c$\label{lin:roundsamplesel}
                \State Let $D^\uns_c$ be a random sample of size $\gamma^\uns_c$ from $B^\uns_c$.\label{lin:roundsampleuns}
             \EndFor
             \State $B^\sel_c \gets \bigcup_c D^\sel_c$
             \State $B^\uns_c \gets \bigcup_c D^\uns_c$
             \State \Return $B^\sel_c, B^\uns_c$
        \EndFunction
        \Function{RoundAndSample}{\textsc{size}, $((S_1, x_1), \ldots (S_k), (x_k)$}
        \State Let $(y_1, \ldots, y_k) = \textsc{size} \cdot (\frac{x_1}{\sum_j x_j}, \ldots, \frac{x_k}{\sum_j x_j})$
        \State Round $(x_1, \ldots, x_k)$ to $(\gamma_1, \ldots, \gamma_k)$ such that $\E[\gamma_i] = x_i$ and $\floor{x_i} \le \gamma_i \le \ceil{x_i}$ and $\sum_i \gamma_i = \textsc{size}$. 
        \For{$i = 1, \ldots, k$}
            \State Let $D_i \subseteq B_i$ be a random sample of size $\gamma_i$\label{lin:round-sampling}
        \EndFor
        \State\Return $\bigcup_i D_i$
        \EndFunction

    \end{algorithmic}
\end{algorithm}

\section{Proof of \Cref{thm:topic}}\label{app:region}

    There are a variety of conditions we must check, namely, that \Cref{alg:region-topic} can successfully run to completion, it returns a valid assembly assignment, and that the returned assignment satisfies our ex ante and ex post guranteess

    Fix an arbitrary rounding $s_{r, t}$ for $r \in R$ and $t \in \topics$ from line~\ref{lin:srt-rounding}. We will take for now that these are fixed constants satisfying $\floor{n \cdot \frac{w_{r, t}}{|N_{r, t}|}} \le s_{r, t} \le \ceil{n \cdot \frac{w_{r, t}}{|N_{r, t}|}} $, only considering randomness from other sampling. In the end, we will deal with the randomness over these $s_{r, t}$ values as well to prove some ex ante guarantees.

    We will show for all $(r, t) \in R \times \topics$, the following properties hold:
    \begin{enumerate}
        \item Each $B^\sel_{r, t}, B^\uns_{r, t} \subseteq N_{(r, t)}$.\label{itm:subset}
        \item Each $|B^\sel_{r, t}| = s_{r, t}$, $|B^\uns_{r, t}| = n - s_{r, t}$. \label{itm:sizes}
        \item Each pair $B^\sel_{r, t}$ and $B^\uns_{r, t}$ are disjoint. \label{itm:pairdisjoint}
        \item All of $(B^\sel_{r, t})_{t \in \topics}$ are pairwise disjoint. \label{itm:alldisjoint}
        \item For $i \in N_{(r, t)}$, $\Pr[i \in B^\sel_{r, t}] = \frac{s_{r, t}}{|L_i| \cdot w_{r, t}}$ and $\Pr[i \in  B^\uns_{r, t}] = \frac{(n - s_{r, t})(1 - 1/|L_i|)}{|N_{r, t}| - w_{r, t}}$.\label{itm:probs}
    \end{enumerate}
    Fix a topic $t \in T$. We will show these properties for all $r$ by structural induction on the graph $R$. We begin by proving it is the case for leaf nodes $r \in \leaves(R)$. Fix such an $r$.

    For such an $r$, we simply need to analyze the $\textsc{SampleLeaves}$ function, showing that it can run successfully and produces an output satisfying the desired properties. 
    
    Properties~\ref{itm:subset}--~\ref{itm:alldisjoint} are immediate from the algorithm assuming it can run to completion. However, we must show that the random choices made on lines~\ref{lin:sample-D} and~\ref{lin:sample-E} are feasible, in the sense that the set we are selecting from has enough people. To that end, fix a set $T \subseteq \topics$, and consider the execution of line~\ref{lin:sample-D}. Note that this is successful as long as $\sum_t \alpha^T_t \le |C^T|$. Indeed, we have $\sum_t \alpha^T_t \le \ceil{\sum_{t \in T} p^T_t}$ by assumption of the rounding. Expanding this, for each $t$,
    \[p^T_t = \frac{|C^T|/|T|}{w_{r, t}} \cdot s_t \le \frac{|C^T|/|T|}{w_{r, t}} \cdot \ceil*{\frac{n \cdot w_{r, t}}{|N_{r, t}|}}.\]
    
    Let $\lambda = w_{r, t} / |N_{r, t}|$. We have that $|N_{r, t}| \ge 4n$, so this is at most 
    \[
        \frac{|C^T|/|T|}{4 \lambda n} \cdot \ceil*{\lambda n}.
    \]
    Now, since $\lambda \ge \varepsilon$, $\lambda n \ge \frac{2 \lambda}{\varepsilon \delta} \ge 1$. Thus $\frac{\ceil*{\lambda n}}{2\lambda n} \le 1$, and we have $p^T_t \le |C^T| / |T|$. Summing up over $t$, we get that $\sum_{t \in T} p^T_t \le |C^T|$. Thus, $|C^T| \ge \ceil{|C^T|} = \ceil{\sum_{t \in T} p^T_t}$ since $|C^T|$ is an integer.

    Next, fix $r$, $t \in \topics$, and $T$ with $t \in T$, and consider the run of line~\ref{lin:sample-E} with these parameters. Note that this sample will be successful as long as $|C^T| \ge \alpha^T_t + \beta^T_t$. We additionally have
    $\alpha^T \le \ceil{p^T_t}$ and $\beta^T_t \le \ceil{q^T_t}$. Therefore, we have
    \begin{align*}
         &\phantom{{}={}}\alpha^T_t + \beta^T_t\\
         &\le \ceil{p^T_t + q^T_t} + 1\\
         &\le \ceil*{\frac{|C^T|/|T|}{w_{r, t}}\ceil*{\frac{n \cdot w_{r, t}}{|N_{r, t}|}} + \frac{|C^T|(1 - 1/|T|)}{|N_{r, t}| - w_{r, t}}\ceil*{\frac{n \cdot (|N_{r, t}| - w_{r, t})}{|N_{r, t}|}}} + 1\\
         &= 
        \Bigg\lceil|C^T|\Bigg((1/|T|) \cdot \left( \frac{1}{w_{r, t}} \cdot \ceil*{\frac{n \cdot w_{r, t}}{|N_{r, t}|}}\right)\\
        &\qquad\qquad\qquad+ (1 - 1/|T|) \left(\frac{1}{|N_{r,t}| - w_{r, t}} \right)\ceil*{\frac{n \cdot (|N_{r, t}| - w_{r, t})}{|N_{r, t}|}} \Bigg)\Bigg\rceil + 1.
    \end{align*}
    Now, since $(1/|T|) + (1 - 1/|T|) = 1$, this is a convex combination between two terms. Thus, we may upper bound it by the maxmium of the two as
    \[
        \le \ceil*{|C^T| \max\left( \frac{1}{w_{r, t}} \cdot \ceil*{\frac{n \cdot w_{r, t}}{|N_{r, t}|}},  \left(\frac{1}{|N_{r,t}| - w_{r, t}} \right)\ceil*{\frac{n \cdot (|N_{r, t}| - w_{r, t})}{|N_{r, t}|}} \right)} + 1.
    \]
    Again, let $\lambda = w_{r, t} / |N_{r, t}| \le w_{r, t}/4n$, we can expand
    \[
        \frac{1}{w_{r, t}} \cdot \ceil*{\frac{n \cdot w_{r, t}}{|N_{r, t}|}} \le \frac{\ceil{\lambda n}}{4n \lambda}.
    \]
    Since $n \lambda \ge \varepsilon \cdot \frac{2}{\varepsilon \delta} \ge 2$, we have that $\frac{\ceil{\lambda n}}{4n \lambda} \le 1/2$. The same argument applies for the second term swapping $\lambda$ with $1 - \lambda$, since we also know $1 - \lambda \ge \varepsilon$. Thus, since we have also assumed $|C^T| \ge 2$, this simplifies to $\ceil{|C^T|/2} + 1 \le |C^T|$, as needed.

    Finally,  we handle property~\ref{itm:probs}. Fix $i \in N_{(r, t)}$ with $L_i = \set{r} \times T$, and note that $t \in T$. Now, by symmetry, each member of $C^T$ is equally likely to be in each of $B^\sel_{r, t}$ and  $B^\uns_{r, t}$. Thus, the probability is simply the expected number of seats going to $C^T$ divided by $|C^T|$. This is precisely $p^T_t / |C^T|$ and $q^T_t / |C^T|$, which expand to our desired values.

    Next, we proceed to the inductive step. Fix an internal node $r \in \feds(R)$ and topic $t \in \topics$. Suppose all of the properties hold $B^\sel_{c, t}$ and $B^\uns_{c, t}$ for $c \in \children(r)$. To show they hold for this $r$ and $t$, we must consider the \textsc{SampleFromChildren} step. We will show it successfully runs, producing $B^\sel_{r, t}$ and $B^\uns_{r, t}$ satisfying all of the properties.

    The two places \textsc{SampleFromChildren} may fail are on lines~\ref{lin:roundsamplesel} and~\ref{lin:roundsampleuns} if we wish to sample more people than are available.
    Line~\ref{lin:roundsamplesel} can run successfully as long as $\gamma^{\sel}_c \le |B^\sel_{c, t}| = s_{c, t}$ for each $c \in \children(r)$. 
    Fix such a child $c$. Note that by the rounding, we have that  $\gamma^{\sel}_c \le \ceil{s \cdot \frac{w_{c, t}}{\sum_{c'} w_{c', t}}}$
    Note that the denominator $\sum_{c'} w_{c', t} = w_{r, t}$ since children are disjoint. Furthermore, we have that
    $s_{r, t} \le \ceil{n \cdot \frac{w_{r, t}}{|N_{r, t}|}}$ and $s_{r, t} \ge \floor{n \cdot \frac{w_{c, t}}{|N_{c, t}|}}$. Therefore, it is sufficient to show
    \[
        \ceil*{\ceil*{n \cdot \frac{w_{r, t}}{|N_{r, t}|}}\cdot \frac{w_{c, t}}{w_{r, t}}} \le \floor*{n \cdot \frac{w_{c, t}}{|N_{c, t}|}}.
    \]
    This is implied by
    \[
        \ceil*{n \cdot \frac{w_{r, t}}{|N_{r, t}|}}\cdot \frac{w_{c, t}}{w_{r, t}} + 1 \le n \cdot \frac{w_{c, t}}{|N_{c, t}|},
    \]
    which itself is implied by
    \[
       (n \cdot \frac{w_{r, t}}{|N_{r, t}|} + 1)\cdot \frac{w_{c, t}}{w_{r, t}} + 1 \le n \cdot \frac{w_{c, t}}{|N_{c, t}|}.
    \]
    Finally, note that 
    \[
        \left(n \cdot \frac{w_{r, t}}{|N_{r, t}|} + 1\right)\cdot \frac{w_{c, t}}{w_{r, t}} + 1  \le n \cdot \frac{w_{r, t}}{|N_{r, t}|} + 2,
    \]
    so this entire inequality follows from
    \[
        n \cdot \frac{w_{r, t}}{|N_{c, t}|} - n \cdot \frac{w_{r, t}}{|N_{r, t}|} \ge 2.
    \] 
    Now, we have that
    \[
    n \cdot \frac{w_{r, t}}{|N_{c, t}|} - n \cdot \frac{w_{r, t}}{|N_{r, t}|} = n \cdot \frac{w_{r, t}}{|N_{c, t}|} (1 - \frac{|N_{c, t}|}{|N_{r, t}|}) \ge n \cdot \varepsilon \cdot \delta \ge 2,
    \]
    by assumption on $n$.

    A similar argument works for line~\ref{lin:roundsampleuns}. Fix $c \in \children(r)$. We need to show that
    \[
        \ceil*{(n - s_{r, t}) \cdot \frac{|N_{c, t}| - w_{c, t}}{\sum_{c'} |N_{c', t}| - w_{c', t}}} \le n - s_{c, t}
    \]
    for all $c \in \children(r)$. Expanding definitions, this is implied by
    \[
        \ceil*{\ceil*{n\left(\frac{|N_{r, t}| - w_{r, t}}{|N_{r, t}|} \right)} \frac{|N_{c, t}| - w_{c, t}}{|N_{r, t}| - w_{r, t}} } \le \floor*{n \cdot \frac{|N_{c, t}| - w_{c, t}}{|N_{c, t}|}}.
    \]
    Doing the same expansion on the ceilings, this inequality is implied by
    \[
        n \cdot \frac{|N_{c, t}| - w_{c, t}}{|N_{r, t}|}  + 2 \le n \cdot \frac{|N_{c, t}| - w_{c, t}}{|N_{c, t}|}.
    \]
    Finally, we have that
    \[
        n \cdot \frac{|N_{c, t}| - w_{c, t}}{|N_{c, t}|} - n \cdot \frac{|N_{c, t}| - w_{c, t}}{|N_{r, t}|} = n \cdot \frac{|N_{c, t}| - w_{c, t}}{|N_{c, t}|} \cdot \left(1 - \frac{|N_{c, t}|}{|N_{r, t}}\right) \ge n \cdot \varepsilon \cdot \delta \ge 2.
    \]

    Now that we have proved the function runs successfully, we prove the properties. We have that $B^\sel_{r, t} \subseteq \bigcup_{c \in \children(r)} B^\sel_{c, t} \subseteq \bigcup_{c \in \children(r)} N_{(c, t)} = N_{(r, t)}$, and a symmetric argument holds for $B^\sel_{r, t}$. The size holds because the sets $B^\sel_{c, t}$ are pairwise disjoint, so we never select the same person twice, and end up with $s_{r, t}$ distinct people. Again, a symmetric argument works for $B^\uns_{r, t}$. The disjointness holds trivially because it held for the children sets, along with distinct child regions having disjoint populations. Finally, fix $i \in N_{(r, t)}$, and suppose $c \in \children(r)$ is the unique child such that $i \in N_c$. Then the only way $i$ can be in $B^\sel_{(r, t)}$ is if $i$ was selected to $B^\sel_{(c, t)}$, and is then subsequently in the subset selected to join $B^\sel_{(r, t)}$. The probability of being in $B^\sel_{(c, t)}$ is $\frac{s_{c, t}}{|L_i| \cdot w_{c, t}}$ by induction. Furthermore, we sample $\gamma^\sel_c$ with $\mathbb{E}[\gamma^\sel_c] = \frac{s_{r, t}}{w_{c, t}}{w_{r, t}}$ and selected $\gamma^\sel_c$ out of the $s_{c, t}$ people uniformly at random from $B^\sel_{(c, t)}$. Thus conditioned on $i \in B^\sel_{(c, t)}$, the probability $i$ is in $B^\sel_{(s, t)}$ is $\frac{s_{r, t}}{w_{c, t}}{w_{r, t}}$. Hence, the overall probability $i \in B^\sel_{(s, t)}$ is
    \[
       \frac{s_{c, t}}{|L_i| \cdot w_{c, t}} \cdot \frac{\frac{s_{r, t}}{w_{c, t}}{w_{r, t}}}{s_{c, t}} = \frac{s_{r, t}}{|L_i| \cdot w_{r, t}},
    \]
    as needed. A symmetric argument holds for $B^\uns_{(r, t)}$.

    Finally, there is one more place where the algorithm can potentially fail, which is in the \textsc{RoundAndSample} call on line~\ref{lin:round-sampling}. To ensure that this line can run, we need to make sure that for each region $r$ and topic $t$, $|B^\sel_{r, t}| = s_{r, t} \le\ceil{n \cdot  \frac{w_{r, t}}{|N_{(r, t)}|}}$ is sufficiently large to be able to sample from. Note that we will take at most $\ceil{n \cdot \frac{w_{r, t}}{\sum_{t'} w_{r, t'}}}$ people. Furthermore, this denominator is $|N_r|$. So, we simply need to show
    \[
        \ceil*{n \cdot \frac{w_{r, t}}{|N_r|}} \le \floor*{n \cdot \frac{w_{r, t}}{|N_{r, t}|}}.
    \]
    Note that since $\floor*{n \cdot \frac{w_{r, t}}{|N_{r, t}|}}$ is an integer, this is implied by
    \[
        n \cdot \frac{w_{r, t}}{|N_r|} \le \floor*{n \cdot \frac{w_{r, t}}{|N_{r, t}|}}.
    \]
    Again, let $\lambda = \frac{w_{r, t}}{|N_{r, t}|}$. We then have
    \begin{align*}
        n \cdot \frac{w_{r, t}}{|N_r|}
        &= n \cdot \lambda \cdot \frac{|N_{r, t}|}{|N_r|}\\
        &\le n \cdot \lambda \cdot (1 - \delta)\\
        &= n \lambda - n\lambda \delta.
    \end{align*}
    Now, since $\lambda \ge \varepsilon$, $n\lambda\delta \ge 2$. Therefore, this is at most \[n\lambda - 2 \le \floor{n\lambda} = \floor*{n \cdot \frac{w_{r, t}}{|N_{r, t}|}},\]
    as needed.

    We have now shown that \Cref{alg:region-topic} can execute to completion. In what we remains we show all of the properties that it satisfies. First, note that each assembly in the final assignment is of size $n$ because it is composed of a disjoint union of sets whose sizes add up to $n$. 

    Next, we show that approximate ex post child representation holds. Fix a region $r$. We first consider $(r, *)$. Note that for each $t$, we select at least $\floor{n \cdot \frac{w_{r, t}}{|N_r|}}$ members from $B^\sel_{(r, t)} \subseteq A_{(r, t)}$. Thus, $A_{(r, t)} \cap A_{(r, *)} \ge \floor{n \cdot \frac{w_{r, t}}{|N_r|}}$, the exact child representation guarantees. Now fix an internal region $r \in \feds(R)$, a child $c \in \children(R)$ and a topic $t \in T$. We show that $|A_{(c, t)} \cap A_{(r, t)}| \ge \floor{n \cdot \frac{N_{(c, t)}}{N_{(r, t)}}} - 1$. Indeed, $|A_{(c, t)} \cap A_{(r, t)}| = |B^\sel_{c, t} \cap B^\sel_{(r, t)}| + |B^\uns_{c, t} \cap B^\uns_{(r, t)}|$. These two sizes are randomly selected on line~\ref{lin:massiverounding} of \textsc{SelectFromChildren}, and their sum will be $\gamma^\sel_c + \gamma^\uns_c$. There is a constraint that \[
    \gamma^\sel_c + \gamma^\uns_c \ge \floor{x^\sel_c + x^\uns_c} = \floor*{s_{r, t} \cdot \frac{w_{c, t}}{w_{r, t} } + (n - s_{r, t}) \cdot \frac{|N_{(c, t)}| - w_{c, t}}{|N_{(r, t)}| - w_{r, t}}}. 
    \]
    We will show that
    \begin{equation}
        s_{r, t} \cdot \frac{w_{c, t}}{w_{r, t} } + (n - s_{r, t}) \cdot \frac{|N_{(c, t)}| - w_{c, t}}{|N_{(r, t)}| - w_{r, t}} \le \frac{|N_{(c, t)}|}{|N_{(r, t)}|} - 1.
    \end{equation}
    which implies the ex post guarantees.
    Note that if $s_{(r, t)}$ were equal to its expectation $n \cdot \frac{w_{r, t}}{|N_{(r, t)}|}$, then the sum simplifies to exactly $\frac{|N_{(c, t)}|}{|N_{(r, t)}|}$. Note that by rounding $s_{(r, t)}$ to a neighboring integer, one of $s_{(r, t)}$ and $(n - s_{r, t})$ is larger than its expectation, and the other is smaller. However, this different is at most one. Furthermore, they are multipled by either $\frac{w_{c, t}}{w_{r, t} } $ or $\frac{|N_{(c, t)}| - w_{c, t}}{|N_{(r, t)}| - w_{r, t}}$, terms which are at most $1$. Hence, even after this rounding, the sum inside can differ from the expectation by at most one, and we get the lower bound of $\frac{|N_{(c, t)}|}{|N_{(r, t)}|} - 1$, as needed.

    Next, we consider ex ante child representation. For a node $(r, *)$, this is straightforward, as we directly round to get in expectation $n \cdot \frac{w_{(r, t)}}{|N_r|}$ selected from $B^\sel_{r, t} \subseteq A_{(r, t)}$. For an internal $r \in \feds(R)$, child $c \in \children(r)$, and topic $t \in \topics$, the overlap of $|A_{(r, t)} \cap A_{(c, t)}|$ is precisely the number of people selected from $B^\sel_{c, t}$ for $B^\sel_{r, t}$ plus the number of people selected from $B^\uns_{c, t}$ for $B^\uns_{r, t}$, as these sets are disjoint. The expected number of each is $s_{(r, t)} \cdot \frac{w_{c, t}}{w_{r, t}}$ and $(n - s_{(r, t)}) \cdot \frac{|N_{(c, t)}| - w_{c, t}}{|N_{(r, t)}| - w_{r, t}}$.  If we take expectation over $s_{(r, t)}$ as well, by linearity, we have that the expected overlap is
    \[
        n \cdot \frac{w_{r, t}}{|N_{(r, t)}|} \cdot \frac{w_{c, t}}{w_{r, t}} + \left(n - n \cdot \frac{w_{r, t}}{|N_{(r, t)}|}\right) \cdot \frac{|N_{(c, t)}| - w_{c, t}}{|N_{(r, t)}| - w_{r, t}} = n \cdot \frac{|N_{(c, t)}|}{|N_{(r, t)}|},
    \]
    as needed.

    Finally, we show individual representation. First, fix a node $(r, *)$ and a person $i \in N_r$ signed up for leaf nodes $\set{r'} \times T$. For each $t \in T$, we have that $i \in B^\sel_{(r, t)}$ with probability $\frac{s_{(r, t)}}{|T| \cdot w_{(r, t)}}$. In expectation, $n \cdot \frac{w_{(r, t)}}{|N_r|}$ will be selected from $B^\sel_{(r, t)}$ to be in $A_{(r, *)}$. hence, conditioned on $i \in  B^\sel_{(r, t)}$, $i$ will be selected with probability $\frac{n \cdot \frac{w_{(r, t)}}{|N_r|}}{s_{r, t}}$. This means the total probability of $i$ being in both $ B^\sel_{(r, t)}$ and $A_{(r, *)}$ is
    $
        \frac{n}{|T| \cdot |N_r|}.
    $
    Note that $i$ being in each of $B^\sel_{(r, t)}$ cannot happen simultaneously, because these sets are disjoint. Therefore, the total probability $i$ is in $A_{(r, *)}$ is the sum of the probabilities of all of these events, and therefore $\frac{n}{|N_r|}$. 

    Finally, consider a node $(r, t)$, and a member $i \in N_{(r, t)}$. We have that the probability of $i \in B^\sel_{(r, t)}$ is $\frac{s_{r, t}}{|L_i| \cdot w_{r, t}}$ and the probability of $i \in B^\uns_{(r, t)}$ is $\frac{(n - s_{r, t})(1 - 1/|L_i|)}{|N_{(r, t)}| - w_{r, t}}$.  Note that these are mutually exclusive events because the sets are disjoint. Hence, the total probability of $i \in A_{(r, t)}$ is the sum
    \[
        \frac{s_{r, t}}{|L_i| \cdot w_{r, t}} + \frac{(n - s_{r, t})(1 - 1/|L_i|)}{|N_{(r, t)}| - w_{r, t}}.
    \]
    By lienarity of expectation over the sampling of $s_{r, t}$, we have that the complete probability is
    \[
        \frac{n \cdot \frac{w_{r, t}}{|N_{(r, t)}|}}{|L_i| \cdot w_{r, t}} + \frac{(n - n \cdot \frac{w_{r, t}}{|N_{(r, t)}|})(1 - 1/|L_i|)}{|N_{(r, t)}| - w_{r, t}} = \frac{n}{|L_i|}{|N_{(r, t)}|} + \frac{n \cdot (1 - 1/|L_i|)}{|N_{(r, t)}|} = \frac{n}{|N_{(r, t)}|},
    \]
    showing individual representation. \qed

\end{document}